\documentclass{jssc}

\def\textsubscript#1%
{$_{\text{#1}}$}

\def\cdd{\mbox{\boldmath$\cdot$}~}

\input{amssym.def}
\include{graphix}

\usepackage{mathtools}
\usepackage{url}
\usepackage[algo2e,ruled,vlined]{algorithm2e}
\usepackage{dsfont}

\usepackage{amsmath}
\usepackage{amscd}
\usepackage{amssymb}
\usepackage{amsfonts}
\usepackage{amsfonts}
\usepackage[plainpages=true,pdfpagelabels,colorlinks=true,citecolor=blue,hypertexnames=false]{hyperref}

\newcommand{\bZ}{{\mathbb{Z}}}
\newcommand{\bQ}{{\mathbb{Q}}}

\newcommand{\bN}{{\mathbb{N}}}
\newcommand{\bT}{{\mathbb{T}}}

\newcommand{\cN}{{\mathcal{N}}}

\newcommand{\im}{\operatorname{im}}

\newcommand{\den}{\operatorname{den}}
\newcommand{\num}{\operatorname{num}}
\newcommand{\resultant}{\operatorname{resultant}}

\makeatletter
\def\@oddfoot{\hfill}
\newcount\shumeicount
\def\setshumei#1#2#3{%
  \shumeicount=\count0
  \def\@oddhead{%
    \raise-5pt\hbox to0pt{\vrule width\hsize height 0pt depth 0.4pt\hss}\relax
    \ifnum \shumeicount=\count0
      \raise-7pt\hbox to0pt{\vrule width\hsize height 0pt depth 0.4pt\hss}\relax
      #1
    \else
      \ifodd\count0
        #2
      \else
        #3
       \fi
     \fi
  }%
}
\makeatother
\makeatletter
\def\@oddfoot{\hfill}
\newcount\shujiaocount
\def\setshujiao{%
  \shujiaocount=\count0
  \def\@oddfoot{%
      \ifodd\count0
      \else
      \fi
  }%
}
\makeatother
\def\biaoti#1#2#3#4{{
  \vspace*{0.3cm}
  \begin{flushleft} \Large\bf #1\end{flushleft}
  \vspace*{-0.2cm}
      \begin{flushleft}
      \bf #2
      \end{flushleft}
      \footnotetext{\hspace{-6mm} #3\\ #4}}}
\def\dshm#1#2#3#4
{\setshumei{J Syst Sci Complex (#1) #2:
{\thepage--\pageref{LastPage}}\hfill}
            {\hfill {\small #3}\hfill\hbox to0pt{\hss\thepage}}
            {\hbox to0pt{\thepage\hss}\hfill {\small #4}\hfill
            }
            \setshujiao}
\def\drd#1#2
{{\vskip 1cm\small \begin{flushleft}
 #1 \\
 #2\\
\copyright The Editorial Office of JSSC \&  Springer-Verlag Berlin
Heidelberg 2014
\end{flushleft}}}
\def\tilde{\widetilde}

\def\epsilon{\varepsilon}

\begin{document}

%*************************************************************************************************************
% \biaoti{THE CAPITALIZED TITLE OF YOUR ARTICLE$^*$}{The list of authors' names with the LAST NAME capitalized
% and the authors' names should be separated by "\cdd"}{the first author's name \\ the first author's affiliation
% and Email address\\ the second author's name\\ the second author's affiliation. More can be listed like this.}
% {$^*$ The titles and numbers of the foundations that support this article.}
%*************************************************************************************************************
\biaoti{Reducing Hyperexponential Functions \\ over Monomial Extensions$^*$} %%%   Main Title of your paper  %%%
{\uppercase{Chen} Shaoshi $\cdd$ \uppercase{Du} Hao $\cdd$ \uppercase{Gao} Yiman $\cdd$ \uppercase{Li} Ziming}%%% The names of the authors  %%%
{\uppercase{Chen} Shaoshi,
KLMM, Academy of Mathematics and Systems Science, Chinese Academy of Sciences, Beijing 100190, China; School of Mathematical Sciences, University of Chinese Academy of Sciences, Beijing 100049, China.
Email\,$:$ schen@amss.ac.cn. \\
\uppercase{Du} Hao,
School of Sciences, Beijing University of Posts and Telecommunications, Beijing 102206, China. Email\,$:$ haodu@bupt.edu.cn.\\
\uppercase{Gao} Yiman (corresponding author),
KLMM, Academy of Mathematics and Systems Science, Chinese Academy of Sciences, Beijing 100190, China; School of Mathematical Sciences, University of Chinese Academy of Sciences, Beijing 100049, China.
Email\,$:$ ymgao@amss.ac.cn. \\
\uppercase{Li} Ziming, 
KLMM, Academy of Mathematics and Systems Science, Chinese Academy of Sciences, Beijing 100190, China; School of Mathematical Sciences, University of Chinese Academy of Sciences, Beijing 100049, China.
Email\,$:$ zmli@mmrc.iss.ac.cn.}
%%% The address of the authors  %%%
{$^*$S. Chen was partially supported by the National Key Research and Development Project 2020YFA0712300, the NSFC grants (No. 12271511 and No. 11688101), CAS Project for Young
Scientists in Basic Research (Grant No. YSBR-034), and the CAS Fund of the Youth Innovation Promotion Association
(No. Y2022001). Hao Du by an NSFC grant (No.\ 12201065),  Y.\ Gao and Z.\ Li by two NSFC grants (No.\ 11971029 and No.\ 12271511).
}
%*************************************************************************************************************
%The submission date of your article. For example: \drd{Received: June 8, 2006}
%*************************************************************************************************************
\drd{DOI: }{Received: x x 20xx}{ / Revised: x x 20xx}

%*************************************************************************************************************
% The page header of the article.
% \dshm{Year}{Volume}{The capitalized RUNNING HEAD of your article with less than 48 letters}{The capitalized
% AUTHORS list with $\cdot$ separating different names or one can type "The name of the first author et al."
% if there are more than 4 authors.}
%*************************************************************************************************************

\dshm{2021}{XX}{Reducing Hyperexponential Functions over Monomial Extensions}{\uppercase{Chen, Shaoshi} $\cdd$ \uppercase{Du, Hao} $\cdd$ \uppercase{Gao, Yiman} $\cdd$ \uppercase{Li, Ziming}}
%*************************************************************************************************************
% \dab{The abstract}{Keywords}
%*************************************************************************************************************
%-------------------------------------------------------------------------
\Abstract{We extend the shell and kernel reductions for hyperexponential functions over the field of rational functions to
a monomial extension. Both of the reductions are incorporated  into one algorithm.
As an application, we present an additive decomposition in rationally hyperexponential towers.
The decomposition yields an alternative algorithm for computing elementary integrals over such towers. 
The alternative can find some elementary integrals that are unevaluated by the integrators in the latest versions of {\sc maple}
and {\sc mathematica}.}      % the abstract

        % the keywords

%\MRSubClass{05B05, 05B25, 20B25}      % MR(2000) Subject Classification

%\baselineskip 15pt

\section{Introduction}

Symbolic integration aims at  developing  algorithms to compute integrals in closed form.
One of its classical topics is to determine whether an integrand has an elementary integral,
and compute such an integral if there exists one.  Fundamental results on this topic are collected and reviewed in~\cite{RaSi2022}.
The monograph~\cite{Bron2005} presents algorithms for integrating transcendental functions.

In symbolic integration, an integrand~$f(x)$ is decomposed in one way or another as 
\begin{equation} \label{EQ:reduce}
 f = \frac{d g }{dx}  + r,
 \end{equation}
where $g$ and $r$ are functions of the same \lq\lq type\rq\rq\  as $f$'s, and $r$ is minimal in
some technical sense. Such a process is referred to as reduction for $f$ in this paper. 
A reduction \eqref{EQ:reduce} yields an additive decomposition $g^\prime + r$ of $f$, provided that $r=0$ if and only if the integral of $f$ is of the same \lq\lq type\rq\rq\ as $f$'s.

Let $C$ be a field of characteristic zero throughout the paper, and $C(x)$ be the field of rational functions in~$x$.
The Hermite-Ostrogradsky reduction in \cite[\S 2.2]{Bron2005} computes $g, r \in C(x)$ such that \eqref{EQ:reduce} holds.
Moreover, $r$ is proper with a squarefree denominator, and  the integral of~$f$ belongs to $C(x)$ if and only if $r$ is equal to zero. In other words, the Hermite-Ostrogradsky reduction
computes an additive decomposition of every element in $C(x)$.
A nonzero function is hyperexponential if its logarithmic derivative belongs to $C(x)$.
Nonzero rational functions are a special instance for hyperexponential functions.
The Hermite reduction in \cite[\S 4.2]{BCCLX2013} computes an additive decomposition of a hyperexponential function.
Reductions do not always yield additive decompositions. For instance, the algorithm {\bf HermiteReduce} in \cite[\S 5.3]{Bron2005} decomposes an integrand
as the sum of a derivative, a simple function and a reduced function.

The goal of this paper is to generalize several reductions for hyperexponential functions over~$C(x)$ to monomial extensions.
We extend and unify the shell and kernel reductions in~\cite{BCCLX2013, GLL2004} (see Theorem \ref{TH:sk}),
and generalize the Hermite reduction in \cite{BCCLX2013} to an additive decomposition algorithm in rationally hyperexponential towers (see Theorem \ref{TH:add}). 
A method is presented in Theorem \ref{TH:elem} for determining elementary integrals over such towers.
\begin{example}  \label{EX:elem}
Let
$$ y = \exp\left( \int \frac{1}{x^3-x-2}\right) \quad 
\text{and} \quad f=\frac{(x^3-x-3) \exp(x) }{(x^3-x-2)\left( \exp(x)+y\right)}.$$
Then
  \[\int f = \log\left(\exp(x)+y \right)-\sum_{\alpha^3 - \alpha - 2 =0} \frac{1}{3 \alpha^2-1} \log(x-\alpha), \]
  which is elementary over $\displaystyle C(x, \exp(x), y)$. However, neither ``int(\,)'' command in {\sc Maple} nor  ``Integrate[\, ]'' command in {\sc Mathematica} finds this closed form.
  We shall show how to find it by Theorems \ref{TH:add} and \ref{TH:elem}. 
\end{example}

The rest of this paper is organized as follows. We describe basic notions in symbolic integration,
and review shell, kernel and polynomial reductions for hyperexponential functions over the field of rational functions in Section \ref{SECT:pre}.
The shell and kernel reductions are extended and unified in Section \ref{SECT:sk}. As a generalization of the Hermite reduction for hyperexponential functions in \cite[\S 4.2]{BCCLX2013},
we present an algorithm for computing additive decompositions of all elements in a rationally hyperexponential tower in  Section \ref{SECT:rh}.

\section{Preliminaries} \label{SECT:pre}

This section consists of two parts. We present basic terminologies for symbolic integration in Section \ref{SUBSECT:basic}, and then recall reductions that will be generalized in Section \ref{SUBSECT:rrh}.

\subsection{Basic definitions and facts} \label{SUBSECT:basic}

In the sequel,  $F$ stands for a field of characteristic zero.  Let $f$ belong to $F(t)$, where $t$ is an indeterminate over $F$.
The numerator and denominator of $f$ are denoted by $\num(f)$ and $\den(f)$, respectively.  They are coprime polynomials in $F[t]$.
In particular, $\num(0)=0$ and $\den(0)=1$.
We say that $f$ is \emph{$t$-proper} if $\deg_t(\num(f))<\deg_t(\den(f))$.
Let $p$ be a  polynomial in $F[t]$ with positive degree.
The \emph{order} of $f$ at $p$ is defined to be $-m$ if $p^m \mid \den(f)$ but $p^{m+1} \nmid \den(f)$ for some $m \in \bZ^+$; while it is defined to be $m$ if $p^m \mid \num(f)$
but $p^{m+1} \nmid \num(f)$ for some $m \in \bN$.
The order is denoted by $\nu_p(f)$.
%We say that $f$ is \emph{free} of $p$ if $\nu_p(f)=0$.

A \emph{derivation} $\delta$ on a field $F$ is an additive map from $F$ to itself satisfying the usual product rule $\delta (ab) = \delta(a) b+a \delta(b)$ for all $a,b \in F$.
The pair $(F,\, \delta)$ is called a \emph{differential field}.
An element of $F$ is called a \emph{constant} if its derivative is zero. The set of constants in $F$ is a subfield of $F$, which is denoted by $C_F$.
An element of $F$ is called a \emph{logarithmic derivative} if it is equal to $\delta(a)/a$ for some $a \in F \setminus \{0\}$. 
Let $(E,\, \Delta)$ and $(F,\, \delta)$ be two differential fields. We call $E$ a {\em differential field extension of $F$}
if $F \subset E$ and $\Delta |_F= \delta$. When there is no confusion, we still denote the derivation $\Delta$ on $E$ by $\delta$.

\noindent 
{\bf  Convention.} \emph{From now on, we let $(F, \, ^\prime)$ be a differential field}, and 
$F^\prime := \{ f^\prime \mid f \in F \}.$

Given an element $f$ of $F$, a reduction decomposes $f$ as $g^\prime + r$, where $g, r \in F$ and $r$ is minimal in some technical sense.
We call  $g^\prime + r$  an \emph{additive decomposition} of $f$,  and $r$ a \emph{remainder} of~$f,$ provided that $r=0$ if and only if $f \in F^\prime$.  
Remainders are not necessarily unique.
Algorithms for computing additive decompositions of all elements in $F$ are particularly useful to determine elementary integrals over $F$, as indicated  in
\cite[Theorem 6.1]{CDL2018} and \cite[Theorem 4.10]{DGLW2020}. 
The Hermite-Ostrogradsky reduction in \cite[\S 2.2]{Bron2005} yields an additive decomposition in $(C(x), d/dx)$. 
Additive decompositions are computed in a finite extension of $C(x)$ in \cite{CKK2016} and in primitive towers of some kinds over $C(x)$ in \cite{CDL2018,DGLW2020}.

Let $E$ be a differential field extension of $F$. An element  $t$ of $E$ is called a \emph{monomial} over~$F$ if it is transcendental over~$F$ and its derivative belongs to $F[t]$.
%In other words, $F[t]$ is a polynomial ring closed under differentiation.
According to \cite[page 7]{Roth1976},
a monomial $t$ over $F$ is said to be \emph{regular} if $C_F = C_{F(t)}$.
A nonzero element of $E$  is said to be \emph{hyperexponential} over $F$ if its logarithmic derivative belongs to $F$.

Let $t$ be a monomial over $F$. For $p \in F[t]$ with $p \neq 0$,  $p$ is \emph{normal} if $\gcd(p, p^\prime)=1$.
It is \emph{special} if $p \mid p^\prime$.  Nonzero elements
of $F$ are both normal and special, and vice versa. They are said to be trivially normal and special.
Basic properties of normal and special polynomials are presented in \cite[\S 3.4]{Bron2005}.

An element $f$ of $F(t)$ is said to be \emph{normally $t$-proper}  if it is $t$-proper and every irreducible factor of $\den(f)$ is normal.
It is said to be \emph{$t$-simple} if it is $t$-proper with a normal denominator,
and it is \emph{$t$-reduced} if $\den(f)$ is special. Clearly, $t$-simple elements are normally $t$-proper and nonzero polynomials in $F[t]$ are $t$-reduced.
\begin{example} \label{EX:ns}
Let $t$ be a regular and hyperexponential monomial over $F$. Then $p \in F[t]$ is special if and only if $p = a t^m$ for some nonzero element $a \in F$ and $m \in \bN$ by \cite[Theorem 5.1.2]{Bron2005}.
The $t$-reduced elements in $F(t)$ are exactly Laurent polynomials in $t$ over $F$.
\end{example}

Let $f \in F(t)$ and $p$ be a nontrivial normal factor of $\den(f)$. By \cite[Theorem 4.4.2 (i)]{Bron2005},
\[  \nu_p(f^\prime) = \nu_p(f) - 1 < \nu_p(f) \le -1. \]
For brevity, the use of this fact will be
referred to as  \emph{an argument on orders} in the sequel. For instance,
$f \notin F(t)^\prime$ for every  nonzero and $t$-simple element $f \in F(t)$ by such an argument.

Let $y$ be hyperexponential over $F$.  Then $y$ can be expressed as
$\exp\left(\int f \right)$,
where $f$ stands for the logarithmic derivative $y^\prime/y$.
Issues related to integrating $y$ include:
(i)   determining whether $y \in F(y)^\prime$,
(ii)  reducing $y$ modulo $F(y)^\prime$,
(iii) developing an additive decomposition in $F(y)$,
(iv)  determining whether $y$ has an elementary integral over $F(y)$.

 Under some additional assumptions, we can replace the additive group $F(y)^\prime$ in reduction algorithms with its subgroup $\left\{ (a y)^\prime \mid a \in F \right\}$, as described in the following proposition.
 \begin{proposition} \label{PROP:similar}
 Let $y$ be a regular and hyperexponential monomial over $F$.
 Then 
 \[  y \in F(y)^\prime \, \Longleftrightarrow \, y = (a y)^\prime \quad \text{for some $a \in F$.} \]
 \end{proposition}
 \begin{proof} By definition, $y \in F(y)^\prime$ if and only if there exists $z \in F(y)$ such that $y=z^\prime$.  
 Then $z$ is $y$-reduced by an argument on orders.
 So $z \in F[y^{-1}, y]$ by Example~\ref{EX:ns}. It follows from the transcendence of $y$ that $z  = a y$ for some $a \in F$. $\Box$ \end{proof}

\subsection{Shell and kernel reductions in $C(x)$} \label{SUBSECT:rrh}
The first three issues listed above are well-handled by the algorithms in \cite{BCCLX2013} when $F=C(x)$.

A rational function $\xi \in C(x)$ is \emph{differential reduced} if
$\gcd(\num(\xi)-i\den(\xi)^\prime, \den(\xi))=1$
for all $i \in \bZ$.  For $f \in C(x)$, there exist $\xi, \eta \in C(x)$ with $\eta \neq 0$ such that
(i) $f = \eta^\prime/\eta + \xi$, (ii) $\xi$ is differential reduced,
(iii)  $\num(\eta)$, $\den(\eta)$ and $\den(\xi)$ are coprime pairwise.

 It is  shown in \cite[\S 3]{GLL2004} that $\xi$ is unique and that $\eta$ is unique up to a multiplicative constant.
 They are called the \emph{kernel} and \emph{shell} of $f$, respectively.

Let $y$ be hyperexponential over $C(x)$, and $\xi, \eta$ be the kernel and shell of $y^\prime/y$, respectively.
Algorithm {\bf  ReduceCert} in \cite[\S 4]{GLL2004} decomposes
\begin{equation} \label{EQ:shell}
\eta  =  u^\prime + u \xi  +  \underbrace{ v  + \frac{p}{\den(\xi)}}_h,
\end{equation}
where $u, v  \in C(x)$,  $v$ is $x$-simple, $\den(v) \mid \den(\eta)$, and $p \in C[x]$, and
$h$ is minimal in the sense that $\den(v)$ divides $\den(\tilde v)$ if there is another reduction
$\eta  =  \tilde{u}^\prime + \tilde{u} \xi   +  \tilde{v}  + \tilde{p}/\den(\xi)$
for some $\tilde{u}, \tilde{v}  \in C(x)$ and $\tilde p \in C[x]$. In terms of hyperexponential functions,
\eqref{EQ:shell} can be expressed as
\[ y = \eta \exp\left(\int \xi   \right) = \left( u \exp\left( \int \xi \right) \right)^\prime + h \exp\left( \int \xi \right),\]
which is a reduction for $y$.
By \cite[Theorem 4]{GLL2004}, $y$ is the derivative of another hyperexponential element if and only if $v=0$ and
$z^\prime + z \xi = {p}/{\den(\xi)}$ for some $z \in C[x]$.
When this is the case, $y = \left(u\eta^{-1}y+z\eta^{-1}y\right)^\prime$.
The equality \eqref{EQ:shell} is called a shell reduction in \cite{BCCLX2013}.

To avoid solving the above differential equation for $z$, we recall the $C$-linear map
\[
\begin{array}{cccl}
\phi_\xi: & C[x]  & \rightarrow & C[x] \\ 
      &  a     &  \mapsto    &  \den(\xi) a^\prime + \num(\xi) a
\end{array}
\]
defined in \cite{BCCLX2013}.
The map is injective, because $\xi$ is differential reduced.  The injectivity allows us to construct a $C$-basis of $\im\left(\phi_\xi \right)$ in a straightforward manner.
This construction also yields a finite-dimensional $C$-linear subspace $\cN$ with $C[x] = \im\left(\phi_\xi\right) \oplus \cN$.
Let $q$ be the projection of $p$ in \eqref{EQ:shell} to $\cN$ with respect to the above direct sum. Then \eqref{EQ:shell} can be refined into
\[    \eta = w^\prime  + w \xi  + \underbrace{v + \frac{q}{\den(\xi)}}_r  \]
for some $w \in C(x)$. 
We call $r$ a {\em residual form} with respect to $\xi$.
Translating the above equality in terms of hyperexponential functions, we have
$
   y  = (w \eta^{-1} y)^\prime + r \eta^{-1} y.
$
By \cite[Lemma 11]{BCCLX2013}, $y$ is the derivative of another hyperexponential element over $C(x)$  if and only if $r = 0$. 
In other words,  $(w \eta^{-1} y)^\prime + r \eta^{-1} y$ is an additive decomposition of $y$. The algorithm for computing this additive decomposition
in \cite{BCCLX2013} is called the Hermite reduction for hyperexponential functions, which, together with the algorithm {\bf HermiteReduce} in \cite[\S 5.2]{Bron2005},
leads to an algorithm for computing additive decompositions in $C(x,y)$ whenever $y$ is a regular monomial over $C(x)$.

Another useful reduction is introduced on the way of computing telescopers for bivariate hyperexponential functions in \cite[\S 6]{BCCLX2013}.
Let $\xi \in C(x)$ be differential reduced. By \cite[Lemma 16]{BCCLX2013}, for every $p \in C[x]$ and $m \ge 1$, we can compute $u \in C(x)$ and $q \in C[x]$ such that
\begin{equation} \label{EQ:kernel}
 \frac{p}{\den(\xi)^m} = u^\prime + u \xi  + \frac{q}{\den(\xi)},
\end{equation}
which is called  a kernel reduction.
\section{Generalizing shell and kernel reductions} \label{SECT:sk}

This section consists of four parts. We generalize the shell reduction \eqref{EQ:shell} and notion of kernels to monomial extensions  in Sections \ref{SUBSECT:shell}
and \ref{SUBSECT:wn}, respectively. The kernel reduction \eqref{EQ:kernel}
is generalized in Section \ref{SUBSECT:kernel}. A generalized kernel-shell reduction in  Section \ref{SUBSECT:unify} unifies the algorithms in  Sections \ref{SUBSECT:shell} and \ref{SUBSECT:kernel}.
%At last, a reduction is presented for hyperexponential elements over a monomial extension in Section \ref{SUBSECT:rhyper}.

Throughout this section,  $t$ stands for a monomial over $F$.
For $f \in F(t)$, we denote by $V_f$ the additive group $\left\{ a^\prime + a f \, | \, a \in F(t) \right\}$.
Let $y$ be hyperexponential over $F(t)$ with logarithmic derivative $f$.
Then an element $a^\prime + a f$ corresponds to $(a y)^\prime$, because  $ \left( a y \right)^\prime = (a^\prime + a f) y. $

\subsection{A generalized shell reduction} \label{SUBSECT:shell}
First, we generalize the shell reduction \eqref{EQ:shell}  to $F(t)$ locally.
\begin{lemma} \label{LM:localshell}
Let $f \in F(t)$, $p \in F[t]$ be normal and coprime with $\den(f)$, and $m \in \bZ^+$. Then,
for every $q \in F[t]$,
there exist $v, w \in F[t]$ with $\deg_t(v) < \deg_t(p)$ such that
\[
 \frac{q}{p^m} \equiv \frac{v}{p} + \frac{w}{\den(f)} \mod V_{f}. \
\]
\end{lemma}
\begin{proof} We set $v =  0 $ and $w= q \den(f)/p^{m} $ if $p \in F$. Set $v$ to be the remainder of $q$ by~$p$, and $w$ to be the product of $\den(f)$ and the quotient of $q$ by $p$ if $m=1$.

Assume that $\deg_t(p)>0$ and $m>1$. Since $p$ is normal,
 there exist $u_1, u_2\in F[t]$
\begin{equation} \label{EQ:bezout}
u_1p +u_2p^\prime=q.
\end{equation}
With the aid of integration by parts, we deduce that
\[\frac{q}{p^m}  = \frac{u_1}{p^{m-1}} + \frac{u_2 p^\prime}{p^m} =  \left(\frac{-(m-1)^{-1} u_2 }{p^{m-1}}\right)^\prime + \frac{(m-1)^{-1} u_2^\prime+u_1}{p^{m-1}}.\]
Setting $u_3 = -(m-1)^{-1} u_2$ and $u_4 = (m-1)^{-1} u_2^\prime+u_1$,  we have
\begin{align*}
\frac{q}{p^m} & = \left( \frac{u_3}{p^{m-1}} \right)^\prime + \frac{u_4}{p^{m-1}} \\
& = \left( \frac{u_3}{p^{m-1}} \right)^\prime +  \frac{u_3}{p^{m-1}} f - \frac{u_3 }{p^{m-1}}  f + \frac{u_4}{p^{m-1}} \\
              & \equiv  - \frac{u_3}{p^{m-1}} f + \frac{u_4}{p^{m-1}} \mod V_f \\
              &  \equiv  \frac{u_5}{p^{m-1}} + \frac{u_6}{\den(f)} \mod V_{f} \quad \text{for some $u_5, u_6 \in F[t]$}. 
\end{align*}
The last congruence is derived by a partial fraction decomposition for $-u_3f/p^{m-1}$ and the assumption $\gcd(p, \den(f))=1$.
Applying the same argument to $u_5/p^{m-1}$ inductively, we get
$$\frac{q}{p^m} \equiv \frac{\tilde{v}}{p} + \frac{\tilde{w}}{\den(f)} \mod V_{f}$$
 for some $\tilde v, \tilde w \in F[t]$.  Let $v$ be the remainder of $\tilde v$ and $p$.
 Then the lemma holds by a similar operation used in the case $m=1$. 
 $\Box$ 
 \end{proof}
 
 A special case of the above lemma plays a key role in Section \ref{SECT:rh}.

  \begin{corollary} \label{COR:localshell}
  With the notation introduced in Lemma~\ref{LM:localshell}, assume further that 
  $$\deg_t(t^\prime) \le 1, \,\, \deg_t(p)>\deg_t(q) \ge 0, \,\, \text{and} \,\, f \in F.$$
  Then $q/p^m \equiv v/p \mod V_f$ for some $v \in F[t]$ with $\deg_t(v)<\deg_t(p)$.
  \end{corollary}
\begin{proof}
Let us go through the second paragraph of the above proof with the additional assumptions in mind.
Since  $\deg_t(t^\prime) \le 1$, we have $\deg_t(p^\prime) \le \deg_t(p)$, which, together with $\deg_t(q) <  \deg_t(p)$,
 implies that $u_1$ and $u_2$ in \eqref{EQ:bezout} can be chosen to have
 degrees less than $\deg_t(p)$. Thus, $u_3$ and $u_4$ given in the proof of the above lemma satisfy the same degree constraints.
 Let~$\tilde u_5 = u_4 - u_3 f$. Then $\tilde u_5$ is a polynomial of degree less than $\deg_t(p)$, because $f \in F$.
 From 
  $$\frac{q}{p^m} \equiv \frac{\tilde{u}_5}{p^{m-1}} \mod V_{f},$$
  and a straightforward induction on $m$, the corollary follows. $\Box$ \end{proof}

\begin{example} \label{EX:Shell}
Let $F=C(x,y)$ with $x'=1$ and $y'=y$, and $t=\exp\left(x^2/2\right)$. Then $t$ is a regular and hyperexponential monomial  over $F$. Let
$f ={y}/{(t-x)}$ and $g={x}/{(t+1)^2}.$
Note that $t \nmid \den(g)$ and $\gcd(\den(g), \den(f))=1$.
So Lemma \ref{LM:localshell} is applicable. Then the algorithm implicitly described in the proof of Lemma~\ref{LM:localshell} yields.
\[
g =\left(\frac{1}{t+1}\right)'+\frac{1}{t+1}f+ \frac{x^2+x+y}{(x+1)(t+1)}-\frac {y}{(x+1)(t-x)}  \equiv \frac{v}{t+1} +  \frac {w}{t-x}\mod V_f,
\] 
where $v = (x^2+x+y)/(x+1)$ and $w = -y/(x+1)$. 
\end{example}

The next lemma helps us prove certain minimality and uniqueness.
\begin{lemma}   \label{LM:unique}
Let $f, h \in F(t)$, $h$ be $t$-simple, and $\gcd(\den(f), \den(h))=1$.  
If there exists a $t$-reduced element $r \in F(t)$ such that $h + r/\den(f) \in V_f$, then $h=0$.
\end{lemma}
\emph{Proof.}  Since $h + r/\den(f) \in V_f$, there exists $a \in F(t)$ such that
\begin{equation} \label{EQ:gosper}
h + \frac{r}{\den(f)} = a^\prime + a f.
\end{equation}
Suppose that $h$ is nonzero. Then $\nu_p(h)=-1$ for some nontrivial normal polynomial $p$, because~$h$ is $t$-simple.
 It follows from $\gcd(\den(h), \den(f))=1$ and \eqref{EQ:gosper} that $\nu_p(a)<0$.  By an argument on orders,  %$\nu_p(a^\prime)<-1$.
 the left and right-hand sides of \eqref{EQ:gosper} have distinct orders at $p$, a contradiction. $\Box$

A generalized shell reduction in $F(t)$ is described in the following proposition.
\begin{proposition} \label{PROP:shell}
Let $f, g \in F(t)$. If $\den(g)$ is free of any nontrivial special factor and  coprime with $\den(f)$,
then there exists a unique $t$-simple $h \in F(t)$ and $q \in F[t]$ such that
\[
 g \equiv h + \frac{q}{\den(f)} \mod V_f  \quad \text{and} \quad \den(h) \mid  \den(g).
\]
\end{proposition}
\begin{proof} Applying Lemma~\ref{LM:localshell} to each $t$-proper component in the squarefree partial fraction decomposition of $g$,  we see that there exists a $t$-simple $h \in F(t)$
and a polynomial $q \in F[t]$ such that $g \equiv  h + q/\den(f)  \mod V_f$,  
and that $\den(h)$ divides $\den(g)$.
The uniqueness of $h$ is immediate from Lemma \ref{LM:unique}.  $\Box$
\end{proof}

\begin{corollary} \label{COR:shell}
 With the notation introduced in Proposition~\ref{PROP:shell}, assume further that $\deg_t(t^\prime) \le 1$, $f \in F$ and that $g$ is normally $t$-proper.
  Then there exists a unique $t$-simple element $h \in F[t]$ such that $g \equiv h \mod V_f$.
\end{corollary}
\begin{proof} Since $f \in F$, we have $\gcd(\den(g), \den(f))=1$. Since $\den(g)$ has no nontrivial special factor,
the above proposition is applicable. Each component in the squarefree partial fraction decomposition of $g$ is of the form
$a/p^m$ for some $a, p \in F[t]$ with $p \mid \den(g)$ and $\deg(a)<\deg(p)$, because $g$ is $t$-proper. 
The corollary follows from Corollary \ref{COR:localshell} and the above proposition.  $\Box$
\end{proof}

For a convenience of later references,
we specify the input and output of the generalized shell reduction. Its pseudo-code can be easily written down according to the proofs of Lemma~\ref{LM:localshell} and Proposition~\ref{PROP:shell}.

\medskip \noindent
{\bf Algorithm GSR} (Generalized Shell Reduction)

\smallskip \noindent
{\bf Input:} a monomial extension $F(t)$ and $f, g \in F(t)$, where $\den(g)$ has no nontrivial special factor and is coprime with $\den(f)$

\smallskip \noindent
{\bf Output:} $(a, h, q)$, where $a, h \in F(t)$ and $q \in F[t]$ such that
$g = a^\prime + a f + h + q/\den(f)$, and $h$ is $t$-simple with $\den(h)|\den(g)$

We can set $q$ in the output of Algorithm {\bf GSR} to be zero if $\deg_t(t^\prime) \le 1$, $f \in F$ and $g$ is normally $t$-proper
by Corollary \ref{COR:shell}.

\subsection{Weak normalization} \label{SUBSECT:wn}

In this subsection, we adapt some results scattered in \cite[\S 6.1]{Bron2005}  to our later use.
An element $f$ of $F(t)$ is said to be \emph{weakly normalized} (resp.\ normalized) if 
$\gcd(\num(f) - i \den(f)^\prime, \den(f))=1$ for all $i \in \bZ^+$ (resp.\  $i \in \bZ$).
Note that $f$ is normalized if and only if $f$ is differential reduced when $F=C$ and $t=x$.
We prefer the word \lq\lq normalized\rq\rq\ rather than the phrase \lq\lq differential reduced\rq\rq, because
the former is concise and compatible with the phrase \lq\lq weakly normalized\rq\rq\ coined in \cite[Definition 6.1.1]{Bron2005}.

The next lemma enables us to generalize the notion of kernels and that of shells from rational functions to elements in a monomial extension.
\begin{lemma} \label{LM:residue}
Let $f \in F(t)$ and $\lambda \in F$. Then $\num(f)-\lambda \den(f)^\prime$ and  $\den(f)$ are not coprime if and only if there exists a nontrivially normal and irreducible polynomial $p$ such that
\begin{equation} \label{EQ:order}
\nu_p(f)=-1 \quad \text{and} \quad \nu_p \left( f - \lambda {p^\prime}/{p} \right) \ge 0.
\end{equation}
\end{lemma}
\begin{proof} Let $p$ be a factor of $\den(f)$. Then $\den(f)=pq$ for some $q \in F[t]$, and 
\begin{equation} \label{EQ:divide}
f - \lambda \frac{p^\prime}{p}  = \frac{\num(f) - \lambda p^\prime q}{\den(f)}.
\end{equation}

Assume that $p$ is nontrivially normal and irreducible, and that it satisfies the constraints in~\eqref{EQ:order}.
By \eqref{EQ:order} and \eqref{EQ:divide},  $p$ divides $\num(f) - \lambda p^\prime q$,
which, together with $\nu_p(f)=-1$, implies that  $p$ is a common factor of $\num(f)-\lambda \den(f)^\prime$ and  $\den(f)$.
Conversely, let $p$ be a nontrivial irreducible factor of  $\gcd\left(\num(f)-\lambda \den(f)^\prime, \den(f)\right)$.
Since $\num(f)-\lambda \den(f)^\prime = \num(f) - \lambda(p^\prime q + p q^\prime),$
we see that $\num(f) - \lambda p^\prime q$ is divisible by $p$, which, together with $p \nmid \num(f)$, implies that $p \nmid p^\prime q$.
Thus $p \nmid p^\prime$ and $p \nmid q$, which imply that $p$ is normal and $\nu_p(f)=-1$, respectively.
The inequality $\nu_p \left( f - \lambda {p^\prime}/{p} \right) \ge 0$ in~\eqref{EQ:order} holds owing to \eqref{EQ:divide}, $\nu_p(f)=-1$ and $p \mid \left(\num(f) - \lambda p^\prime q\right)$.  $\Box$
\end{proof}

We present an algorithm to construct $\xi, \eta \in F(t)$ with $\eta \neq 0$ for a given $f \in F(t)$
such that (i) $f = \eta^\prime/ \eta + \xi,$ (ii) $\xi$ is weakly normalized (resp.\ normalized),
(iii) $\den(\eta)$  is free of any nontrivial special factor,
(iv) $\gcd(\den(\eta), \den(\xi))=1$ and $\gcd(\num(\eta), \den(\xi))=1$.

\medskip \noindent
{\bf Algorithm GKS} (Generalized Kernel and Shell)

\smallskip \noindent
{\bf Input:} a monomial extension $F(t)$ and $f \in F(t)$

\small \noindent
{\bf Output:} $\xi, \eta \in F(t)$ satisfy the four requirements listed above
\begin{itemize}
\item[(1)]  {\bf if} $f \in F$ {\bf then} {\bf return} $f, 1$
\item[(2)]  $\xi \leftarrow f$ and $\eta \leftarrow 1$
\item[(3)]  $g \leftarrow $ the product of the normal and irreducible factors of $\den(\xi)$ with multiplicity 1
\item[(4)]  factor $g$ over $F$ to get its nontrivial irreducible factors: $g_1, \ldots, g_k$
\item[(5)]  {\bf for} $i$ {\bf from} $1$ {\bf to} $k$ {\bf do}
\item[]
\begin{itemize}
 \setlength\leftskip{0.5em} 
 \item[(5.1)] $p \leftarrow \num\left(\xi - z g_i^\prime/g_i\right),$ where $z$ is a constant indeterminate
\item[(5.2)]  $r \leftarrow$ the remainder of $p$ by $g_i$
\item[(5.3)] set $r=0$ to obtain a system of linear equations in $z$ over $F$.
\item[(5.4)] {\bf if} the system has a solution $m \in \bZ^+$ (resp.\ $m \in \bZ$) {\bf then}
                  \begin{itemize}
                   \item[] $\eta \leftarrow \eta g_i^{m}$ and $\xi \leftarrow \xi  - m g_i^\prime/g_i$
                   \end{itemize}
                  {\bf end if}
\end{itemize}
{\bf end do}
\item[(6)] {\bf return} $\xi, \eta$
\end{itemize}
%}}

%\vspace{0.3cm}
\noindent
Note that $g$ in step (3) of Algorithm {\bf GKS} can be found as follows.  Compute
$$w = \frac{\den(f)}{\gcd(\den(f),\den(f)^\prime)},$$
which is  the product of all normal and irreducible factors of $\den(f)$ by \cite[Lemma 3.4.4]{Bron2005}. Then
$g$ is equal to $w/\gcd(w, \den(f)^\prime)$
by a straightforward calculation.
The correctness of the algorithm then follows from Lemma \ref{LM:residue}. We call $\xi$ and $\eta$ computed by Algorithm {\bf GKS}$(F(t),f)$ the
{\em weakly normalized (resp.\ normalized) kernel and the corresponding shell in $F(t)$}, respectively.
Searching for $m \in \bZ^+$ in step (5.3) finds a weakly normalized kernel;
while looking for $m \in \bZ$, we get a normalized one.

\begin{example}
Let $F=C(x)$ and $t$ be a regular and hyperexponential monomial over $F$ with $t'/t = 1/(x^2+1)$. Applying Algorithm {\bf GKS} to 
$$f = \frac{x^3t+x^2t+2xt+t+1}{(xt+1)(x^2+1)}$$
yields the  weakly normalized kernel  $1/(x^2+1)$
 and shell $xt+1$. 
\end{example}
\subsection{A generalized kernel reduction} \label{SUBSECT:kernel}
We extend the kernel reduction~\eqref{EQ:kernel} to $F(t)$.
\begin{proposition}  \label{PROP:kernel}
Let $f \in F(t)$ be weakly normalized. Then, for every $p \in F[t]$ and a positive  integer $m$, there exists $q \in F[t]$ such that
${p}/{\den(f)^m}  \equiv {q}/{\den(f)}  \mod V_f.$
\end{proposition}
\begin{proof}  If $m=1$ or $p=0$, then set $q=p$. Otherwise, there exist $u, v \in F[t]$ such that
\[      u \left(\num(f)-(m-1)\den(f)^\prime\right) + v \den(f) = p \]
since $f$ is weakly normalized.
Using integration by parts, we have
\[
  	 \frac{p}{\den(f)^m}  = \left(\frac{u}{\den(f)^{m-1}}\right)^\prime +  \left(\frac{u}{\den(f)^{m-1}}\right) f + \frac{v-u^\prime}{\den(f)^{m-1}} 
	                                 \equiv  \frac{v-u^\prime}{\den(f)^{m-1}}  \mod V_f.
\]
The proposition then follows from a straightforward induction on $m$.   $\Box$
\end{proof}

The algorithm described in the above proof is specified below.

\medskip \noindent 
{\bf Algorithm GKR} (Generalized Kernel Reduction)

\smallskip \noindent
{\bf Input:} a monomial extension $F(t)$, a weakly normalized element $f \in F(t)$, a polynomial $p \in F[t]$ and a positive integer $m$

\smallskip \noindent
{\bf Output:} $a \in F(t)$ and $q \in F[t]$ such that
${p}/{\den(f)^m} = a^\prime + a f + {q}/{\den(f)}$

\begin{example}\label{EX:Kernel}
Let $F(t)$ and $f$ be given in Example \ref{EX:Shell} and 
$g=(y+1-xt)/{(t-x)^2}.$
Since $f$ is weakly normalized,  Algorithm {\bf GKR} yields $g=\left( {1}/(t-x) \right)'+  {f}/(t-x) \equiv 0 \mod V_f$.
\end{example}
\begin{example} 
Let $F=C(x,y)$ with $x'=1$ and $y'=xy$, and let $t=\exp(y)$. Then $t$ is a regular and hyperexponential monomial  over $F$. Let
$$f =\frac{1}{t+y} \quad \text{and} \quad  g=\frac{(y+1-x^2y)t-x^2y+y^2+x+y}{(t+y)^2}.$$
Note that $f$ is weakly normalized and $\den(g)=\den(f)^2$.
Algorithm {\bf GKR} yields $$g=\left(\frac{x}{t+y} \right)'+ \left(\frac{x}{t+y}\right)f +\frac {y}{t+y}\equiv\frac {y}{t+y} \mod V_f.$$
\end{example}

\subsection{A generalized kernel-shell reduction} \label{SUBSECT:unify}

We are ready to extend and unify shell and kernel reductions.
\begin{theorem}  \label{TH:sk}
Let $f, g \in F(t)$ and $f$ be weakly normalized.
Then the following assertions hold.
\begin{itemize}
\item[(i)] There exists a unique $t$-simple element $h$ with
$\den(h) \mid \den(g)$ and $\gcd(\den(h), \den(f))=1,$
and a $t$-reduced element $r$ such that
\begin{equation} \label{EQ:red}
      g  \equiv h + \frac{r}{\den(f)} \mod V_f.
\end{equation}
\item[(ii)] If $g  \equiv \tilde h + \tilde{r}/\den(f) \mod V_f,$
where $\tilde h \in F(t)$ and $\tilde r$ is $t$-reduced,
then $\den(h) \mid \den(\tilde h)$.
\item[(iii)] $g \in V_f$ if and only if $h=0$ and there exists a $t$-reduced element $a \in F(t)$ such that
$$ \frac{r}{\den(f)} = a^\prime + f a.$$
\item[(iv)] Assume further that  $\deg_t(t^\prime) \le 1$, $f \in F$, and that $g$ is normally $t$-proper. Then \eqref{EQ:red} can be rewritten as $g \equiv h \mod V_f$. 
Moreover, $g \in V_f$ if and only if $h=0$.
\end{itemize}
\end{theorem}
\emph{Proof.}  By a partial fraction decomposition for $g$, we have
\begin{equation} \label{EQ:decomp}
g = g_1 + g_2 + g_3,
\end{equation}
where $g_1, g_2, g_3 \in F(t)$, $g_1$ and $g_2$ are $t$-proper, all the irreducible factors of $\den(g_1)$ are normal and coprime with $\den(f)$, those of
$\den(g_2)$ are factors of $\den(f)$, and those of $\den(g_3)$ are special and coprime with $\den(f)$.

(i) By Proposition \ref{PROP:shell},  there exists a $t$-simple element $h$ with $\den(h)|\den(g_1)$, and $q_1 \in F[t]$ such that
$g_1 \equiv h + q_1/\den(f) \mod V_f.$
Note that $g_2$ can be written as $p/\den(f)^m$ for some $p \in F[t]$ and $m \in \bZ^+$. By Proposition \ref{PROP:kernel}, there exists $q_2 \in F[t]$
such that $g_2 \equiv q_2/\den(f) \mod V_f.$
Since $g_3$ is $t$-reduced,  the above two congruences and~\eqref{EQ:decomp} lead to
$g \equiv h + r/\den(f)  \mod V_f,$
where $r = q_1 + q_2 + g_3 \den(f)$.
The uniqueness of $h$ is evident by Lemma \ref{LM:unique}.

(ii) By (i),  
$\tilde h \equiv h^* + r^*/\den(f) \mod V_f$
with $\den(h^*)|\den(\tilde h)$ and $\gcd(\den(h^*), \den(f))=1$ for some $t$-simple element $h^*$ and $t$-reduced element $r^*$.
  Therefore,   
$g \equiv h^* + (\tilde{r} + r^*)/\den(f) \mod V_f.$
Since $\tilde r + r^*$ is $t$-reduced, we have $h = h^*$ by (i). Consequently, $\den(h) \mid \den(\tilde h)$.

(iii) By \eqref{EQ:red},  $g \in V_f$ if $h=0$ and $r / \den(f) \in V_f$.
Conversely, assume that $g \in V_f$.   Then $h=0$ by~(ii). It follows from \eqref{EQ:red} that there exists $a \in F(t)$ such that 
\begin{equation} \label{EQ:reduced}
r  = a^\prime \den(f) + a \num(f).
\end{equation}
It remains to show that $a$ is $t$-reduced.
Suppose that $p$ is a nontrivial irreducible and normal polynomial with $m:=\nu_p(a)<0$.
Then $\nu_p(a^\prime)=m-1$. 
It follows from \eqref{EQ:reduced}, $\nu_p(r) \ge 0$ and an argument on orders that neither $\nu_p(f) \ge 0$ nor $\nu_p(f)<-1$. 
So $\nu_p(f)=-1$. Consequently, $\nu_p( a \den(f)) \le 0$ so that the order of $r/(a \den(f))$ at $p$ is nonnegative.
Hence, \eqref{EQ:reduced} implies $\nu_p(f + a^\prime/a) \ge 0$.
By the logarithmic derivative identity, there exist  $n_i \in \bZ$ and $q_i \in F[t]$ with $\gcd(p, q_i)=1$ such that
$$ f + \frac{a^\prime}{a} = f + m \frac{p^\prime}{p} + \sum_i n_i \frac{q_i^\prime}{q_i}, $$
which, together with $\nu_p\left(f + a^\prime/a\right) \ge 0$ and  $\nu_p\left(\sum_i n_i q_i^\prime/q_i\right) \ge 0$,
implies that $\nu_p\left(f + m p^\prime/p\right) \ge 0.$
Then  $\num(f) + m  \den(f)^\prime$ and  $\den(f)$ are not coprime by Lemma \ref{LM:residue}. Since $m$ is a negative integer,
$f$ is not weakly normalized, a contradiction.

(iv) If $\deg_t(t^\prime) \le 1$, $f \in F$ and $g$ is normally $t$-proper,
then both $g_2$ and $g_3$ in~\eqref{EQ:decomp} are equal to zero.  By Corollary \ref{COR:shell},  $g \equiv h \mod V_f$ for some $t$-simple $h \in F(t)$.
The other conclusion holds by (ii). $\Box$

Based on Theorem~\ref{TH:sk} and its proof, we present a generalized kernel-shell reduction.

\medskip \noindent
{\bf Algorithm GKSR} (Generalized Kernel-Shell Reduction)%for reducing hyperexponential functions}

\smallskip \noindent
{\bf Input:}  a monomial extension $F(t)$, a weakly normalized element $f \in F(t)$ and $g \in F(t)$

\smallskip \noindent
{\bf Output:} $a, h, r \in F(t)$ with $h$ being $t$-simple and $r$ being $t$-reduced such that
 $$ \den(h)|\den(g),  \quad \gcd(\den(h), \den(f))=1 \quad \text{and} \quad g=a^\prime + a f  + h + \frac{r}{\den(f)}$$
\begin{enumerate}
\item[(1)] use a partial fraction decomposition to compute $g_1, g_2, g_3 \in F(t)$ such that
$$g = g_1 + g_2 + g_3,$$
where $g_1$ is normally $t$-proper with $\gcd(\den(g_1), \den(f)){=}1$, $g_2$ is $t$-proper, every irreducible factor of $\den(g_2)$ divides $\den(f)$, and $g_3$ is $t$-reduced
\item[(2)]  $(a_1, h, r_1) \leftarrow$ Algorithm {\bf GSR}$(F(t), f, g_1)$
\item[(3)]  find $p \in F[t]$ and the minimal $m \in \bZ^+$ such that $g_2=p/\den(f)^m$, \\ and
            $(a_2, q)$ $\leftarrow$ Algorithm {\bf GKR}$(F(t), f, p, m)$
\item[(4)]  $(a,r) \leftarrow \left(a_1+a_2, \, r_1+q + g_3 \den(f) \right)$
\item[(5)] {\bf  return} $a, h, r$
\end{enumerate}

\noindent
The correctness of Algorithm {\bf GKSR} is immediate from Propositions~\ref{PROP:shell}, \ref{PROP:kernel} and Theorem \ref{TH:sk} (i).
The $t$-reduced element $r$ is zero if  $\deg_t(t^\prime) \le 1$, $f \in F$ and $g$ is normally $t$-proper by Theorem \ref{TH:sk} (iv). 

%\subsection{A reduction algorithm for hyperexponential elements} \label{SUBSECT:rhyper}

Let $y$ be hyperexponential over $F(t)$. By Algorithm {\bf GKS}, we compute the weakly normalized kernel $\xi$ and the corresponding shell $\eta$ of $y^\prime/y$ in $F(t)$.
Set $z=y/\eta$. Then $z^\prime/z = \xi$.  Let us reduce $g z$ for an element $g \in F(t)$.
By Algorithm {\bf GKSR}, we compute an element $u \in F(t)$, a $t$-simple element~$h$ and a $t$-reduced element $r$ such that
$ g = u^\prime + u \xi + h + {r}/{\den(\xi)}.$
It follows that
\[ g z = (u z)^\prime + \left(h + \frac{r}{\den(\xi)}\right) z. \] 
This is a reduction  for $gz$ for all $g \in F(t)$. In particular, setting $g=\eta$ yields a reduction for $y$.

\begin{corollary} \label{COR:add}
Let $\deg_t(t^\prime) \le 1$,  $y$ be a regular and hyperexponential monomial over $F(t)$ with~$y^\prime/y \in F$, and $g \in F(t)$ be nonzero and normally $t$-proper.
Then there exists a unique $t$-simple element $h \in F(t)$ such that $g y = (uy)^\prime + h y$, which is an additive decomposition of $gy$. 
\end{corollary}
\begin{proof}
Since $y^\prime / y \in F$, its kernel and shell in $F(t)$ are $y^\prime/y$ and $1$, respectively.
There exists a unique $t$-simple element $h \in F(t)$ such that $g y = (uy)^\prime + h y$ by Theorem \ref{TH:sk} (iv). 
Assume $g y \in F(t,y)^\prime$.
By Proposition~\ref{PROP:similar}, there exists $v \in F$ such that $g y = (vy)^\prime$.
Consequently, $g \in V_f$, which, together with Theorem \ref{TH:sk} (iv),  implies
that $h = 0$. $\Box$
\end{proof}
%The reduction algorithm for hyperexponential elements over $F(t)$ is an additive decomposition if the conditions in the above corollary are satisfied.

\begin{example}\label{EX:ShellandKernel}
%Let $F=\bC(x,y)$ with $x'=1$ and $y'/y=1$, and $t=\exp(x^2/2)$ be  regular and hyperexponential  over $F$. Let
Let $F(t)$ and $f$ be given by Example \ref{EX:Shell}. Consider
$$g=\frac{-xt^3+(y-x+1)t^2+(2y-2x^2-x+2)t+x^3+y+1}{(1+t)^2(t-x)^2} \in F(t)$$
Since $f$ is weakly normalized, its kernel is $f$ and the shell is $1$.
First,  we decompose $g$ as
\[g = \underbrace{\frac{x}{(1+t)^2}}_{g_1}+\underbrace{\frac{y+1-xt}{(t-x)^2}}_{g_2},\]
where $g_1$ is normally $t$-proper with $\gcd(\den(g_1), \den(f))=1$, the  irreducible factor $t-x$ of  $\den(g_2)$ divides $\den(f)$.
By Algorithm {\bf GKSR}, Examples \ref{EX:Shell} and \ref{EX:Kernel}, 
\begin{align*}
g & =  \left( \frac{1}{t-x}+ \frac{1}{1+t}\right)'+\left( \frac{1}{t-x}+ \frac{1}{1+t}\right)f+\frac{x^2+x+y}{(x+1)(1+t)}+\frac {-y}{(x+1)(t-x)} \\
  & \equiv   \underbrace{\frac{x^2+x+y}{(x+1)(1+t)}}_h+ \frac {-y}{(x+1)(t-x)} \mod V_f.
\end{align*}
In other words,
\[g \underbrace{\exp\left(\int f \right)}_z  = \left(  \left( \frac{1}{t-x}+ \frac{1}{1+t} \right) z  \right)^\prime + \left(h + \frac {-y}{(x+1)(t-x)} \right) z.
\]
 Since $h \neq 0$, we have that $gz \notin F(t,z)^\prime$ by Theorem \ref{TH:sk} (ii).
\end{example}

\section{An additive decomposition in rationally hyperexponential towers} \label{SECT:rh}

This section has four parts.  In Section \ref{SUBSECT:lm}, we present a variant of the Matryoshka decomposition in \cite{DGLW2020}.
An algorithm is developed for computing additive decompositions in rationally hyperexponential towers in Section \ref{SUBSECT:rh}. 
We describe the projections of logarithmic derivatives in terms of residues, and present a criterion on elementary integrability over such towers in Sections \ref{SUBSECT:logder} and \ref{SUBSECT:elem}, respectively.
 \subsection{Laurent-Matryoshka decompositions} \label{SUBSECT:lm}
 For $n \in \bZ^+$, we denote $\{1,2,\ldots,n\}$ and $\{0,1,2,\ldots,n\}$ by $[n]$ and $[n]_0$, respectively. Let $F_0$ be a  field. For every $i \in [n]$, we further let $F_i=F_{i-1}(t_i)$, where $t_i$ is transcendental over $F_{i-1}$. 
 Then there is a chain of field extensions:
 \begin{equation} \label{EQ:tower}
\begin{array}{ccccccc}
F_0 & \subset & F_1 & \subset & \cdots & \subset & F_n \\
&  & \shortparallel &  & &   & \shortparallel \\
& & F_0(t_1) & \subset & \cdots & \subset & F_{n-1}(t_n).
\end{array}
\end{equation}
For each $i\in [n]$, $f\in F_n$ is said to be \emph{$t_i$-proper} if $f \in F_i$ and $\deg_{t_i}(\num(f))<\deg_{t_i}(\den(f))$.
By a power product of $t_1, \ldots, t_n$, we mean the product $t_1^{\ell_1} \cdots t_n^{\ell_n}$, where the $\ell_i$'s are integers. 
For all~$i \in [n-1]_0$, we denote by $\bT_i$ the set of power products of $t_{i+1}, \ldots, t_n$, and set $\bT_n=\{1\}$.

In the rest of this paper, we let $F_0$ be a  differential field.
Assume that each generator $t_i$ is regular and hyperexponential over $F_{i-1}$. We call \eqref{EQ:tower} a {\em hyperexponential tower}.
For $i \in [n]$,  an element $f$ of~$F_n$  is {\em $t_i$-simple} if it is $t_i$-proper and $\den(f)$ is normal as an element of $F_{i-1}[t_i]$, and
it is {\em $t_i$-reduced} if it belongs to $F_{i-1}[t_i^{-1}, t_i]$ (see Example \ref{EX:ns}). Similarly, $f$ is normally $t_i$-proper if $f$ is $t_i$-proper and $t_i$ does not divide $\den(f)$ in $F_{i-1}[t_i]$.
Zero is normally $t_i$-proper for all $i \in [n]$.
An element of $F_i$ can be written uniquely as the sum of a normally $t_i$-proper element and an element of $F_{i-1}[t_i^{-1}, t_i]$.

Let $L_n$ be the additive subgroup consisting of all normally $t_n$-proper elements in $F_n$. For $i \in [n-1]$, let $L_i$ be the additive group generated by
elements of the form  $a T$, where $a \in F_i$ is normally $t_i$-proper and $T \in \bT_i$.
%Then $L_1, \ldots, L_n$ are all $F_0$-linear subspaces.
Moreover, let $L_0$ be the ring of Laurent polynomials in $t_1, \ldots, t_n$ over $F_0$.
Then $F_n = L_0 \oplus L_1 \oplus \cdots \oplus L_n$ %\bigoplus_{i=0}^n L_i$$ 
by a straightforward verification.   
Let $\pi_i$ be the projection from $F_n$ onto $L_i$ with respect to the above direct sum for all $i\in [n]_0$.
For $f\in F_n$, 
$f = \pi_0(f) + \pi_1(f) + \cdots + \pi_n(f)$ % \sum_{i=0}^n \pi_i(f)$$
is called the \emph{Laurent-Matryoshka decomposition}  of $f$.

\begin{example}
  Let $F_0=\bQ(x)$. A Laurent-Matryoshka decomposition in $F_3$ is 
$$\underbrace{\frac{t_2t_3(x-t_3)}{t_1(t_2+1)(t_3-1)}}_f = \underbrace{-t_3t_1^{-1}+(x-1)t_1^{-1}}_{\pi_0(f)} \, + \, \underbrace{0}_{\pi_1(f)} \, + \, \underbrace{\frac{t_3}{t_1(t_2+1)}- \frac{x-1}{t_1(t_2+1)}}_{\pi_2(f)} \, + \, \underbrace{\frac{(x-1)t_2}{t_1(t_2+1)(t_3-1)}}_{\pi_{3}(f)}.$$
\end{example}

\subsection{Rationally hyperexponential towers} \label{SUBSECT:rh}

Rationally hyperexponential towers are hyperexponential towers of a special type. They allow us to
apply the Hermite reduction in~\cite{BCCLX2013} and additive decomposition  in Corollary \ref{COR:add} directly.
\begin{definition}\label{def:flat}
The tower $F_n$ in \eqref{EQ:tower} is said to be \emph{rationally hyperexponential} if  $t_i^\prime/t_i \in F_0$ for every $i \in [n]$, $(F_0, \, ^\prime)=(C(x),d/dx)$ and $C_{F_n}=C$.
\end{definition}
\begin{lemma} \label{LM:der}
Let the tower $F_n$ in \eqref{EQ:tower} be rationally hyperexponential and $g$ belong to $F_n$.
If 
\begin{equation} \label{EQ:form}
g = \sum_{i \in [n]_0} \underbrace{\sum_{T \in \bT_i} g_T T}_{\pi_i(g)},
\end{equation}
where the coefficient $g_T$ in $\pi_0(g)$ belongs to $F_0$ and $g_{T}$ in $\pi_i(g)$ with $i \in [n]$ is normally $t_i$-proper,
then  
\begin{equation} \label{EQ:dform}
g^\prime = \sum_{i \in [n]_0}  \underbrace{ \sum_{T \in \bT_i}\left(g_T^\prime +\frac{T^\prime}{T} g_T\right) T}_{\pi_i\left( g^\prime\right)}.
\end{equation}
Consequently, $\pi_i(g)^\prime = \pi_i(g^\prime)$ for all $i \in [n]_0$. 
\end{lemma}
\begin{proof} Since $T^\prime/T \in F_0$,
we have that $g_T^\prime + \left(T^\prime/T\right) g_T$ belongs to $F_0$ if $g_T$ is a coefficient in $\pi_0(g)$,  and it is normally $t_i$-proper if $g_T$ is a coefficient in $\pi_i(g)$ for $i \in [n]$.
Therefore,  the lemma follows from the identity that $\left(g_T T \right)^\prime = \left(g_T^\prime + \left(T^\prime/T \right) g_T \right)T$. $\Box$
\end{proof}

We need some notation to describe remainders in a rationally hyperexponential tower $F_n$ in \eqref{EQ:tower}. For~$i \in [n]$, set 
$R_i$ to be the additive group generated by $\{ h T \mid \text{$h \in F_i$ is $t_i$-simple and $T \in \bT_i$} \}.$
For~$T \in \bT_0$, we let $\xi_T$ be the normalized kernel and $\eta_T$ the corresponding shell of $T^\prime/T$ in $F_0$. 
Set $R_0$ to be the additive group generated by 
$$\left\{ r \left( \eta_T^{-1} T \right) \mid  \text{$T \in \bT_0 \setminus \{1\}$ and $r$ is a residual form w.r.t.\ $\xi_T$} \right\} \cup \left\{ s \mid \text{$s \in F_0$ is $x$-simple} \right\}.
$$
Finally, we let $R= \sum_{i \in [n]_0} R_i$, which is a direct sum by the observation that $R_i \subset L_i$ for all $i \in [n]_0$.
\begin{theorem} \label{TH:add}
With the notation just introduced, we have $F_n = F_n^\prime \oplus R.$
\end{theorem}
\begin{proof} 
First, we show that $F_n = F_n^\prime + R$. It suffices to show that $f T \in F_n^\prime + R_i$ for every  normally $t_i$-proper element $f$, $T \in \bT_i$ and $i \in [n]$,
and that $f T \in F_n^\prime  + R_0$ for all $f \in F_0$ and $T \in \bT_0$.

Let $i \in [n]$. By Corollary \ref{COR:add}, there exists a $t_i$-simple $h$ such that $f T \equiv h T \mod F_n^\prime$, where $T$ is regarded as a hyperexponential element over $F_i$. 
Thus, $f T \in F_n^\prime + R_i$ by the definition of $R_i$. 

We regard each $T \in \bT_0 \setminus \{1\}$ as a hyperexponential element over $F_0$.
Let $\xi_T$ and $\eta_T$ be the normalized kernel and shell of $T^\prime/T$ in $F_0$, respectively. For $f \in F_0$,
a partial fraction decomposition for $f \eta_T$ yields $f \eta_T  = a + b$
with $\gcd(\den(a), \den(\xi_T))=1$ and $\den(b) \mid \den(\xi_T)^m$ for some $m \in \bN$.
Then there exists an $x$-simple element $h \in C(x)$ and $u, v \in C[x]$ such that
\[ a \left(\eta_T^{-1} T\right)    \equiv \left( h + \frac{u}{\den(\xi_T)} \right)   \left(\eta_T^{-1} T\right) \mod F_n^\prime \quad 
\text{and} \quad b \left(\eta_T^{-1} T\right)    \equiv \frac{v}{\den(\xi_T)}\left(\eta_T^{-1} T\right)  \mod F_n^\prime \]
by the shell and kernel reductions in \cite{BCCLX2013}, respectively. It follows that
\[  f T  \equiv \left( h + \frac{u+v}{\den(\xi_T)} \right) \left(\eta_T^{-1} T\right)  \mod F_n^\prime. \]
The polynomial reduction in \cite{BCCLX2013} finds a polynomial $p \in C[x]$ such that
\[  f T \equiv \underbrace{\left( h + \frac{p}{\den(\xi_T)} \right)}_r  \left(\eta_T^{-1} T\right) \mod F_n^\prime, \]
where $r$ is a residual form with respect to $\xi_T$. 
So  $f T \in F_n^\prime+R_0$ for all $f \in F_0$ and $T \in \bT_0 \setminus \{1\}$. 
In addition,  there exists an $x$-simple element $s \in F_0$ such that $f \equiv s \mod F_n^\prime$ by the Hermite-Ostrogradsky reduction. 
Therefore, $f T \in F_n^\prime + R_0$ for all $f \in F_0$ and $T \in \bT_0$. Consequently, $F_n = F_n^\prime +R$.

It remains to show that $F_n^\prime \cap R = \{0\}$. 
For $f \in F_n^\prime \cap R$, there exists $g \in F_n$ such that $f = g^\prime$.
Let the Laurent-Matryoshka decomposition of $g$ be given in \eqref{EQ:form}. 
Then the  Laurent-Matryoshka decomposition of $f$ is given in \eqref{EQ:dform}
by Lemma~\ref{LM:der}.
On the other hand, $f \in R$ implies that for all $i \in [n]$,  $$\pi_i(f) = \sum_{T \in \bT_i} r_T T,$$
where $r_T$ is $t_i$-simple. It follows that $r_T = g_T^\prime + (T^\prime/T) g_T$
for all $T \in \bT_i$ and $i \in [n]$. Consequently, $g_T=r_T = 0$ by an argument on orders, that is, $g \in L_0$ and $f \in R_0$ with $f = g^\prime$.
For $T \in \bT_0 \setminus \{1\}$, 
\[ \sum_{T \in \bT_0} r_T \left(\eta_T^{-1} T\right)  = \sum_{T \in \bT_0} \left(g_T^\prime + \frac{T^\prime}{T} g_T \right) T, \]
where $\xi_T$ and $\eta_T$ are the  same as above, $r_T$ is a residual form with respect to $\xi_T$, and $g_T \in F_0$. 
So~$r_T\left(\eta_T^{-1} T\right)  = (g_T T)^\prime$.
By \cite[Lemma 11]{BCCLX2013}, $r_T=g_T=0$. 
Accordingly, $f, g \in F_0$ and $f = g^\prime$. We have that $f$ is $x$-simple by $f \in R_0$. Hence, $f=0$, that is,  $F_n^\prime + R$ is a direct sum.  $\Box$
\end{proof}
By the above theorem, for every element $f \in F_n$, there exists $g \in F_n$ and a unique $r \in R$ such that $f = g^\prime + r$.
In other words, $g^\prime + r$ is an additive decomposition of $f$. Moreover, the remainder $r$ is unique due to the direct sum $F_n^\prime \oplus R$.

\medskip \noindent 
{\bf Algorithm AD\_RHT} (Additive Decomposition in a Rationally Hyperexponential Tower)

\smallskip \noindent
\text{\bf Input:} a rationally hyperexponential tower $F_n=C(x)(t_1,t_2,\ldots,t_n)$ and $f\in F_n$ with $f \neq 0$

\smallskip \noindent
\text{\bf Output:}  $g,r \in F_n$ such that $g^\prime + r$ is an additive decomposition of  $f$

\begin{itemize}
  \item[(1)]  $(g, r, 0) \leftarrow$ Algorithm {\bf GKSR}$(F_n, 0, \pi_n(f))$
  \item[(2)]{\bf for} $i$ {\bf from} 1 {\bf to} $n-1$  {\bf do}
\begin{itemize}
 \setlength\leftskip{0.5em} \item[(2.1)]   write $\pi_i(f)=\sum_{j \in J}  a_j T_j,$ where $a_j \in F_i \setminus \{0\}$ and $T_j \in \bT_i$
\item[(2.2)] {\bf for} each $j \in J$ {\bf do}
  \quad \begin{itemize}
   \item[] $(u_j, v_j, 0)$ $\leftarrow$ Algorithm {\bf GKSR}$(F_i, T_j^\prime/T_j, a_j)$ and $(g, r)$ $\leftarrow$ $(g+u_j T_j, r+v_j T_j)$
  \end{itemize}
  {\bf end do}
\end{itemize}
\item[] {\bf end do}
\item[(3)] write
          $\pi_0(f)=s + \sum_{j \in J} a_j T_j,$
          where $s, a_j \in C(x)$ with $a_j \neq 0$,  and $T_j \in \bT_0 \setminus \{1\}$
\item[(4)]  find $u, v \in C(x)$ such that $s = u^\prime + v$ with $v$ being $x$-simple by the Hermite-Ostrogradsky reduction, and $(g,r) \leftarrow (g+u,r+v)$
          \end{itemize}
\begin{itemize}
\item[(5)] {\bf for} each $j \in J$ {\bf do}
\begin{itemize}
           \item[] compute $g_j, r_j \in F_n$ such that $a_j T = g_j^\prime + r_j$ by the Hermite reduction in \cite{BCCLX2013} \\
                   $(g, r) \leftarrow (g+g_j, r+r_j)$
           \end{itemize}
           {\bf end do}
\item[(6)] {\bf return} $g,r$
\end{itemize}
The correctness of Algorithm {\bf AD\_RHT} is immediate from Corollary \ref{COR:add} and the paragraphes  for establishing $F_n=F_n^\prime  + R$ in the proof of Theorem \ref{TH:add}.

\begin{example} \label{EX:add1}
Find an additive decomposition of
$$f=-\frac{\exp(x)(x-1)}{\exp(x^2/2)}+\frac{\exp(-1/x)}{(1+\exp(x^2/2))^2}+\frac{x}{(\exp(-1/x)+x)^2} $$
in the rationally hyperexponential tower
$$F_3=C(x)\big(\underbrace{\exp(x)}_{t_1}, \, \underbrace{\exp\left(x^2/2\right)}_{t_2}, \, \underbrace{\exp(-1/x)}_{t_3}\big).$$
In the tower $F_3$, $f=-(x-1)t_1t_2^{-1}+t_3/(1+t_2)^2+x/(t_3+x)^2.$  Algorithm {\bf AD\_RHT} yields 
$$f = \left(-\frac{x^2}{(x-1)(t_3+x)}+\frac{t_3}{x(1+t_2)}+t_1t_2^{-1}\right)'+\underbrace{\frac{(x^3+x-1)t_3}{x^3(1+t_2)}+\frac{x^2-3x+1}{(x-1)^2(t_3+x)}}_r.$$
We conclude $f \notin F_3^\prime$ by $r \neq 0$.
\end{example}
\begin{example}\label{EX:add2}
Let $t_1=\exp(x),\, t_2=y$, where $y$ is given in Example \ref{EX:elem}, and let $f$ be the same as that in Example \ref{EX:elem}.
By Algorithm {\bf AD\_RHT}, 
\begin{equation}\label{EQ:add2}
  f=\left(\frac{1}{t_2 + 1}\right)'+\underbrace{\frac{(x^3-x-3)t_1}{(x^3-x-2)(t_1+t_2)}}_r.
\end{equation}
 So $f \notin F_2^\prime$, because $r \neq 0$.
\end{example}
\subsection{Logarithmic derivatives in rationally hyperexponential towers} \label{SUBSECT:logder}
The notion and basic properties of residues  are described in \cite[\S 4.4]{Bron2005}.
Let $f$ be a nonzero element of $F(t)$, and $\den(f)$ be nontrivially normal.
Then the nonzero residues of $f$ are exactly the roots of its Rothstein-Trager resultant by \cite[Theorem 4.4.3]{Bron2005}.
Residues are closely related to elementary integrals according to  \cite[Theorem 3.1]{Raab2012a}.

\begin{example} \label{EX:log}
Let $p \in F[t]$ be a normal polynomial of positive degree $d$. Then the Rothstein-Trager resultant of $p^\prime/p$ is equal to 
$(-1)^d \resultant_t (p^\prime, p) (z-1)^d.$
Thus,  all nonzero residues of $p^\prime/p$ are equal to $1$.  It follows from the logarithmic derivative identity in \cite[Theorem 3.1.1]{Bron2005} that the residues of a logarithmic derivative in $F(t)$
are integers.
\end{example}

We are going to describe the residues of the projections of logarithmic derivatives in a rationally hyperexponential tower.
\begin{lemma} \label{LM:log}
 Let $t$ be a hyperexponential monomial over $F$, and $p \in F[t]$ be normal. Then
 \[
   \frac{p^\prime}{p} = \deg_t(p) \frac{t^\prime}{t} + \frac{a^\prime}{a} + h
 \]
 for some $a \in F$ and some $t$-simple $h \in F(t)$.  Moreover, all the nonzero residues of $h$ are equal to $1$.
 \end{lemma}
 \begin{proof}
 Let $p=a t^m + q$, where $m>0$, $a \in F\setminus \{0\}$ and $q \in F[t]$ with degree lower than $m$.
 Then 
 $p^\prime = \left( a^\prime + m a {t^\prime}/{t} \right) t^m  + q^\prime$ and $\deg \left(q^\prime\right)<m$.
 So 
 $    {p^\prime}/{p} =  m {t^\prime}/{t} + {a^\prime}/{a}   + {r}/{p}, $
 where $r$ is the remainder  of $p^\prime$ and $p$ with respect to $t$.   Setting $h=r/p$ proves the first conclusion.  
By \cite[Theorem~4.4.1]{Bron2005}, taking residues is $F$-linear. Therefore, the residues of $p^\prime/p$ are equal to those of $h$, because both $t^\prime/t$ and $a^\prime/a$ are free of $t$.
The second conclusion then follows from Example \ref{EX:log}. $\Box$
\end{proof}
\begin{proposition} \label{PROP:logder}
Let  the tower $F_n$ in \eqref{EQ:tower} be rationally hyperexponential,  and let $f \in F_n$ be nonzero. Then $\pi_0(f^\prime/f) \in F_0$ and 
$\pi_i(f^\prime/f)$ is a $t_i$-simple element with integral residues for all $i \in [n]$.
\end{proposition}
\begin{proof} Let $i \in [n]$, and let $p \in F_{i-1}[t_i]$ be nontrivially normal with respect to $t_i$.

\smallskip \noindent
\emph{Claim.}  The projection $\pi_i\left( p^\prime/p \right)$ is $t_i$-simple, and
$p^\prime/p = s + a^\prime/a  + \pi_i\left(p^\prime/p \right)$
for some $s \in F_0$ and~$a \in F_{i-1}$. Moreover, all nonzero residues of $\pi_i\left( p^\prime/p \right)$ are equal to $1$.

\smallskip \noindent
\emph{Proof of the claim.} By Lemma \ref{LM:log}, there exists an element  $a \in F_{i-1}$ and a $t_i$-simple element $q \in F_i$ such that
$p^\prime/p = \deg_{t_i}(p) t_i^\prime/t_i + a^\prime/a  + q,$
and that all nonzero residues of $q$ are equal to $1$. Since $t_i^\prime/t_i \in F_0$, we can set $s= \deg_{t_i}(p) t_i^\prime/t_i$.
By the definition of Laurent-Matryoshka decompositions,  we have that $q = \pi_i(p^\prime/p)$.  The claim is proved.

Based on the claim, we proceed by induction on $n$. For $n=1$, the logarithmic derivative identity implies 
${f^\prime}/{f} = f_0 + \sum_{j \in [k]} m_j {p_j^\prime}/{p_j},$ where $f_0 \in F_0$, $p_j$ is a nontrivially normal polynomial in~$F_0[t_1]$, and $m_j$ is a nonzero integer. 
It follows that $\pi_1 \left(  {f^\prime}/{f} \right) =  \sum_{j \in [k]} m_j \pi_1 \left({p_j^\prime}/{p_j}\right).$
The claim implies that~$\pi_1 \left(f^\prime/f\right)$ is $t_1$-simple and its residues belong to $\bZ$. 
The claim also implies that
$\pi_0(f^\prime/f)$ is equal to $f_0 + u + {v^\prime}/{v}$
for some $u, v \in F_0$. So $\pi_0(f^\prime/f)$ belongs to $F_0$.

Assume that $n>1$ and the lemma holds for $n-1$.  Again,  the logarithmic derivative identity implies   
${f^\prime}/{f} = u + {g^\prime}/{g} + \sum_{j \in [k]}m_j {p_j^\prime}/{p_j},$
where $u \in F_0$, $g \in F_{n-1}$, $p_j \in F_{n-1}[t_n]$ is nontrivially normal, and $m_j \in \bZ$. Then
$\pi_n \left(  {f^\prime}/{f} \right) =  \sum_{j \in [k]} m_j \pi_n \left({p_j^\prime}/{p_j}\right),$
which is $t_n$-simple and has only integral residues by the claim.
Moreover, there exists $v \in F_0$ and $h \in F_{n-1}$ such that 
$f^\prime/f=v+h^\prime/h + \pi_n(f^\prime/f)$. 
By the induction hypothesis,  $\pi_\ell (h^\prime/h)$ is $t_\ell$-simple and has only integral residues, and $\pi_0\left(h^\prime/h\right) \in F_0$.
The induction is completed by the observation that $\pi_\ell(f^\prime/f)= \pi_\ell(h^\prime/h)$ for all $\ell \in [n-1]$. $\Box$
\end{proof}

\subsection{Computing elementary integrals} \label{SUBSECT:elem}
An element $f$ of $F$ has an elementary integral over $F$ if
there exists an elementary extension $E$ of~$F$ such that $f \in E^\prime$.
With the aid of remainders and residues, we determine whether $f$ has an elementary integral over a rationally hyperexponential tower.

\begin{theorem} \label{TH:elem}
Let $F_n$ in \eqref{EQ:tower} be a rationally hyperexponential tower, and $f \in F_n$ with the remainder $r$.
Then $f$ has an elementary integral over $F_n$ if and only if $\pi_0(r) \in F_0$ is $x$-simple, and, for all $i \in [n]$,
 $\pi_i(r) \in F_i$ is a $t_i$-simple element whose residues are in the algebraic closure $\overline{C}$ of $C$.
\end{theorem}
\begin{proof} Assume that $\pi_0(r) \in F_0$ is $x$-simple,  $\pi_i(r) \in F_i$ is $t_i$-simple, and all the residues of $\pi_i(r)$ belong to $\overline{C}$ for all $i \in [n]$. 
Then $\pi_i(r)$ is the sum of a $\overline C$-linear combination of logarithmic derivatives
and a polynomial $u$ in $F_{i-1}[t_i]$ by \cite[Lemma 3.1(i)]{DGGL2023}. Since $t_i^\prime/t_i \in F_0$, the polynomial $u$ belongs to~$F_0$ by the expression for $u$ in the proof
of \cite[Lemma 3.1(i)]{DGGL2023}.
It follows that $\pi_i(r)$
has an elementary integral over $F_i$ for all $i \in [n]$,
which, together with $\pi_0(r) \in F_0$, implies that $r$ has an elementary integral over~$F_n$, and so does $f$.

To show the converse, we assume that $f$ has an elementary integral over $F_n$.
Then $r$ also has an elementary integral over the same tower. By \cite[Theorem 5.5.3]{Bron2005},
there exist $g, u_1, \ldots, u_k \in F_n$ and $c_1, \ldots, c_k \in \overline{C}$ such that
$$
r = g^\prime + \sum_{j \in [k]} c_j \frac{u_j^\prime}{u_j}.
$$
Accordingly, for each $i \in [n]_0$,
\begin{equation} \label{EQ:liouville}
\pi_i(r) = \pi_i(g^\prime) + \sum_{j \in [k]}  c_j \pi_i \left(\frac{u_{j}^\prime}{u_{j}} \right).
\end{equation}

First, we show that $\pi_0(r)$ is $x$-simple. By Proposition \ref{PROP:logder}, 
$\pi_0 \left({u_{j}^\prime}/{u_{j}}\right)$ in \eqref{EQ:liouville} belongs to $F_0$ for all~$j \in [k]$.
So $\pi_0(r) \equiv w  \mod F_n^\prime$ for some $w \in F_0$ by $\pi_0(g^\prime)=\pi_0(g)^\prime$ in Lemma \ref{LM:der}.
We may further assume that~$w$ is $x$-simple by the Hermite-Ostrogradsky reduction.
%Since $r \in R$ and $R = R_0 \oplus \cdots \oplus R_n$, we have that $\pi_0(r) \in R_0 $. 
Therefore, $\pi_0(r)-w \in R_0 \cap F_n^\prime$.
By Theorem~\ref{TH:add}, we have that $\pi_0(r) = w$. Consequently, $\pi_0(r)$ is $x$-simple.
%Hence, $\pi_0(r) \in F_0+F_n^\prime$ by Proposition \ref{PROP:logder}.
%Then there exists an $x$-simple element $w$ such that $\pi_0(r) - w \in F_n^\prime$ by the Ostrogradsky-Hermite reduction.
%Since $r \in R$ and $R = R_0 \oplus \cdots \oplus R_n$, we have that $\pi_0(r) \in R_0$. So $\pi_0(r)-w \in R_0$.
%By Theorem~\ref{TH:add}, $\pi_0(r) = w$. Consequently, $\pi_0(r)$ is $x$-simple.

It remains to show that $\pi_i(r)$ is $t_i$-simple and has only constant residues for all $i \in [n]$.
By Proposition~\ref{PROP:logder}, $\pi_i \left({u_{j}^\prime}/{u_{j}} \right)$ in \eqref{EQ:liouville} is $t_i$-simple for all $j \in [k]$.  
Since the coefficient of every power of $t_{i+1}, \ldots, t_n$ in $\pi_i(r)$ is $t_i$-simple,
we have that $\pi_i(g^\prime)=0$ by $\pi_i(g^\prime)=\pi_i(g)^\prime$ in Lemma~\ref{LM:der} and an argument on orders. 
Then $\pi_i(r) = \sum_{j \in [k]}  c_j \pi_i \left({u_{j}^\prime}/{u_{j}}\right)$ by  \eqref{EQ:liouville}.
Consequently, $\pi_i(r)$ is a $t_i$-simple element whose residues belong to $\overline{C}$ by Proposition \ref{PROP:logder}.  $\Box$
\end{proof}
Algorithms for determining constant residues are given in \cite[\S 5.6]{Bron2005} and \cite{Raab2012a,DGGL2023}.

\begin{example}
  Let us reconsider the tower $F_3$ and the function $f$ in Example \ref{EX:add1}. Note that
  $$\pi_2(r)= \frac{x^3+x-1}{x^3(t_2+ 1)}t_3,$$ which is not $t_2$-simple. So  $f$ has no elementary integral over $F_3$ by Theorem \ref{TH:elem}.
\end{example}
\begin{example}
  Let us reconsider the tower $F_2$ and the function $f$, which are the same as those in Examples \ref{EX:elem} and \ref{EX:add2}. By \eqref{EQ:add2}, the remainder $r$ of $f$ has only one nonzero projection, which is 
  $$\pi_2(r)=\frac{(x^3-x-3)t_1}{(x^3-x-2)(t_1+t_2)}.$$
  It is $t_2$-simple. An algorithm for determining constant residues yields
  $$ \pi_2(r) =\frac{(t_1+t_2)'}{t_1+t_2}-\frac{1}{x^3-x-2}.$$
  The nonzero residues of $\pi_2(r)$ are all equal to $1$.  So $f$ has an elementary integral over $F_2$, which is 
  $$\int f  = \frac{1}{1+t_2}+\log\left(t_1+t_2\right)- \sum_{\alpha^3 - \alpha - 2 =0} \frac{1}{3 \alpha^2-1} \log(x-\alpha).$$
\end{example}
%\section{Summary and future research}
%
%We have developed two algorithms {\bf RH} and {\bf AD}. The former extends both the shell and kernel reductions in $C(x)$ to a regular monomial extension $F(t)$.
%It provides a reduction on hyperexponential elements over $F(t)$ modulo $F(t)^\prime$. The latter provides an additive decomposition for elements in a rationally-hyperexponential tower.
%Moreover, a new method is presented for determining the elementary integrability over such towers by computing the residues of remainders in the additive decomposition.
%
%A more challenging problem is to develop an additive decomposition in hyperexponential towers, which are not rationally hyperexponential. Algorithm {\bf RH} appears indispensable
%for tackling this problem.

%\bibliographystyle{plain}

%\bibliography{scrfs}

\end{document}